\documentclass[conference]{IEEEtran}
\IEEEoverridecommandlockouts
\usepackage{cite}
\usepackage{amsmath,amssymb,amsfonts}
\usepackage{algorithmic}
\usepackage{graphicx}

\usepackage{algorithm}
\usepackage{algorithmic}
\usepackage{subfig}

\usepackage{amsthm}
\newtheorem{theorem}{Theorem}

\newtheorem{solution}{Solution}

\makeatletter
\newcommand{\linebreakand}{%
  \end{@IEEEauthorhalign}
  \hfill\mbox{}\par
  \mbox{}\hfill\begin{@IEEEauthorhalign}
}
\makeatother

\usepackage{textcomp}
\usepackage{xcolor}
\def\BibTeX{{\rm B\kern-.05em{\sc i\kern-.025em b}\kern-.08em
    T\kern-.1667em\lower.7ex\hbox{E}\kern-.125emX}}
\begin{document}

\title{Communication-Efficient Collaborative LLM Inference via Distributed Speculative Decoding\\
}

\author{
\IEEEauthorblockN{1\textsuperscript{st} Ce ZHENG}
\IEEEauthorblockA{\textit{Department of Network Intelligence} \\
\textit{Pengcheng Laboratory}\\
Shenzhen, China \\
ce.zheng@pcl.ac.cn}
\and
\IEEEauthorblockN{2\textsuperscript{nd} Tingting YANG}
\IEEEauthorblockA{\textit{Department of Network Intelligence} \\
\textit{Pengcheng Laboratory}\\
Shenzhen, China \\
yangtt@pcl.ac.cn}
}

\maketitle

\begin{abstract}
Speculative decoding is an emerging technique that accelerates large language model (LLM) inference by allowing a smaller “draft” model to predict multiple tokens in advance, which are then verified or corrected by a larger “target” model. In AI-native radio access networks (AI-RAN), this paradigm is well-suited for collaborative inference between resource-constrained end devices and more capable edge servers or base stations (BSs). However, existing distributed speculative decoding requires transmitting the full vocabulary probability distribution from the draft model on the device to the target model at the BS, which leads to prohibitive uplink communication overhead. To address this issue, we propose a ``Top-$K$ Sparse Logits Transmission (TK-SLT)'' scheme, where the draft model transmits only the top-$K$ token raw probabilities and the corresponding token indices instead of the entire distribution. This approach significantly reduces bandwidth consumption while maintaining inference performance. We further derive an analytical expression for the optimal draft length that maximizes inference throughput, and provide a theoretical analysis of the achievable speedup ratio under TK-SLT. Experimental results validate both the efficiency and effectiveness of the proposed method.

\end{abstract}


\section{Introduction}


The advent of \textbf{large language models (LLMs)} marks a paradigm shift in artificial intelligence, enabling applications from generative interfaces to autonomous coding assistants. Yet their deployment remains challenging across cloud, edge, and AI-native radio access networks (AI-RANs). Edge devices face tight memory, energy, and compute limits, while cloud inference suffers from latency, jitter, and mobility-induced disconnections. In AI-RANs, where communication and computation are jointly managed, efficient resource orchestration is essential to deliver reliable, low-latency LLM services.

To address these challenges, researchers have proposed a collaborative edge-device architecture that strategically deploys a small language model (SLM) on the device while offloading the large language model (LLM) to a base station (BS) or edge server~\cite{ding2024hybrid, hao2024hybrid}. In~\cite{ding2024hybrid}, a router trained to predict query difficulty and desired quality level enables cost-efficient assignment of queries to either the small or large model. In~\cite{hao2024hybrid}, a cost-aware draft-verification approach was employed. By tuning a predefined threshold $p_t$ for the generated token probability, a controllable performance-cost trade-off was achieved.


Previous approaches typically enhanced efficiency at the expense of inference accuracy. To address this, \textbf{speculative decoding (SD)} has been proposed, in which a lightweight \emph{draft} model autoregressively predicts $\gamma$ candidate tokens, and a larger \emph{target} model validates them in parallel~\cite{leviathan2023fast, chen2023accelerating}. This mechanism alleviates the inefficiency of sequential token generation while maintaining output quality. Building on this idea, \textbf{distributed speculative decoding (DSD)} was introduced, where the draft model operates on the device and the target model performs verification at BS~\cite{zhao2024edge,oh2024uncertainty,ning2025dssd}. For instance,~\cite{zhao2024edge} explores optimizing the number of tokens produced by the smaller model to jointly reduce delay and energy consumption under uplink and downlink transmission constraints.

Nevertheless, this hybrid deployment is hindered by communication overhead: each token requires transmitting its full probability distribution to the BS for verification, leading to a payload that grows linearly with vocabulary size. For instance, a 32k vocabulary with FP16 can incur about 500 kbit per token. To mitigate this issue,~\cite{oh2024uncertainty} introduced an uncertainty-aware scheme that omits uplink transmission for tokens with high acceptance likelihood at the BS. While this improves throughput, it introduces extra uncertainty estimation overhead and degrades inference accuracy.

This calls for a communication-efficient solution without compromising inference quality. We propose a scheme that applies top-$K$ sampling at the SLM so that only the logits of the top-$K$ candidates are transmitted to the BS, substantially reducing communication overhead while preserving exact equivalence of the resulting inference distribution to that of a standalone LLM. We term this scheme ``Top-$K$ Sparse Logits Transmission (TK-SLT)''.

The rest of the paper is organized as follows: Section~\ref{sec:System model} describes the system model. 
Section~\ref{sec:top-K} presents our TK-SLT scheme.
Section~\ref{sec:theoretical_analysis} provides theoretical analysis, including the Speedup ratio and the closed-form expression for the optimal draft token length. Section~\ref{sec:experiment} reports our numerical analysis and simulation results. Section~\ref{sec:conclusion} concludes the paper..

\section{System Model}\label{sec:System model}
We consider the distributed speculative decoding (DSD) framework, where a lightweight small language model (SLM) is deployed on the device, while a large language model (LLM) is hosted at the BS~\cite{zhao2024edge, oh2024uncertainty}. Both SLM and LLM are assumed to share a common vocabulary $\mathcal{V}$, which includes all possible tokens.


\begin{figure*}[htbp]
\begin{center}
\centerline{\includegraphics[width=0.75\textwidth]{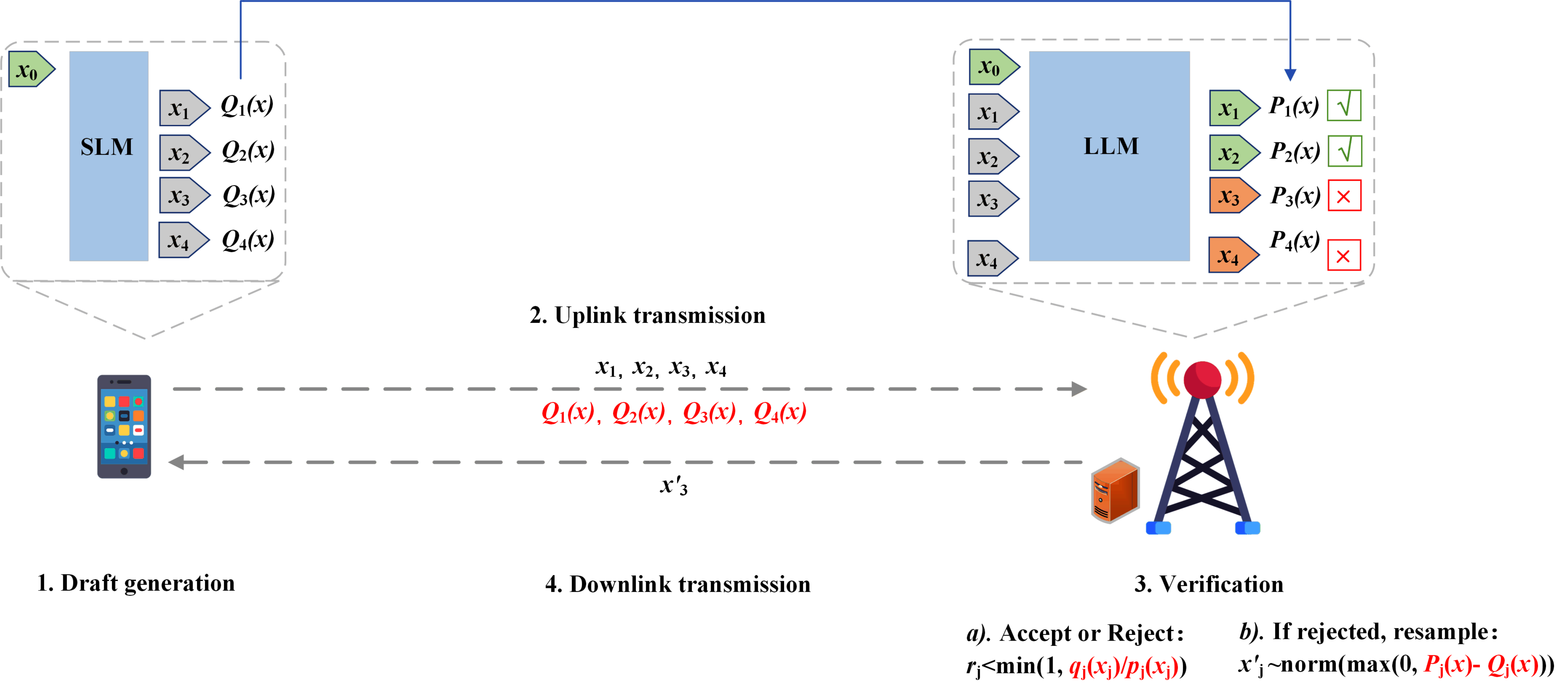}}
\caption{Distributed Speculative Decoding}
\label{fig:DSD}
\end{center}
\end{figure*}

\subsection{Distributed Speculative Decoding}
\label{sec:DSD}

As illustrated in Fig.~\ref{fig:DSD}: SLM generates $\gamma$ draft tokens autoregressively. And LLM verifies them in parallel, accepting valid tokens and resampling new ones for rejected tokens:

\noindent\textbf{1.~Draft Process (on Deivce)}: \\
SLM generates $\gamma$ tokens based on the \textit{prefix}. Specifically,, for the $i$-th token, SLM first gets a vocabulary probability distribution, denoted as $Q_i(x)$ and then samples $x_i$ according to $Q_i(x)$, i.e., $x_i \sim Q_i(x)$.\footnote{$Q_i(x)$ and $P_i(x)$ are vectors with the same dimension as the vocabulary, i.e., $|\mathcal{V}|$.}

\noindent\textbf{2.~Uplink Transmission}: \\
The device sends the indices of the draft tokens and vocabulary probability distributions to the BS.\footnote{Instead of transmitting full token strings, the device sends only the indices of draft tokens from the vocabulary, reducing communication overhead. For the sake of better illustration, however, tokens are consistently used in Fig.~\ref{fig:DSD}.}

\noindent\textbf{3.~Verification Process (at the BS)}: \\
LLM first obtains the $\gamma+1$ distributions based on the prefix and received tokens: $P_j(x), j=1,\cdots,\gamma+1$. It then verifies these received tokens as follows: \\
\textbf{\textit{a).~Accept/Reject:}} Let $q_j(x_j)$ and $p_j(x_j)$ denote the probability values of token $x_j$ in $Q_j(x)$ and $P_j(x)$. If $p_j(x_j)<q_j(x_j)$, $x_j$ is accepted as the $j$-th token. Otherwise, it is rejected with probability of $1-q_j(x_j)/p_j(x_j)$. \\
\textbf{\textit{b).~Resample:}} Once rejected, it resamples a new token $x'_j$ from an adjusted distribution $P'_j(x)=\mathrm{norm} \left(\max \left\{ 0, P_j(x)-Q_j(x) \right\}\right)$. 
If all $\gamma$ received tokens are accepted, it samples the $(\gamma+1)$-\textit{th} token $x_{\gamma+1} \sim P_{\gamma+1}(x)$.

\noindent\textbf{4.~Downlink Transmission}: \\
BS sends the results $x_j$ and $j$ back to the device, where $j$ denotes the position of the resampled token in the sequence if there is a rejection or equals $\gamma+1$ if all draft tokens are accepted.

\subsection{Wireless Communication}
The communication time consists of the uplink transmission time and downlink transmission time:
\begin{equation}
    T_{comm} = T_{up} + T_{down}.
\end{equation}
where
$T_{up} = \frac{D_{up}} {R_{up}}$ and $T_{down} = \frac{D_{down}}{R_{down}}$
are the amount of data transmitted in uplink and downlink, and $R_{up}$ and $R_{down}$ are the transmission rate.

Given that the index size is insignificant relative to the vocabulary distribution, our analysis considers only the uplink transmission latency associated with the vocabulary distribution. And
\begin{align}
    \label{eq:Vocabulary}
    D_{up} = \gamma \cdot D_{\mathcal{V}}, ~~~D_{\mathcal{V}} = |\mathcal{V}| \cdot b_{prob},
\end{align}
where $D_{\mathcal{V}}$ is the amount of data for a single vocabulary distribution, $|\cdot|$ denotes the cardinality, and $b_{prob}$ represents the bit-width of each probability value, e.g. $b_{prob}=32$ bits for full precision or 16 bits for half precision. 

Hence, we have
\begin{equation}
    T_{comm} = \frac{D_{up}} {R_{up}} = \gamma \cdot T_{\mathcal{V}},
\end{equation}
where 
$T_{\mathcal{V}} =  \frac{|\mathcal{V}| \cdot b_{prob}} {R_{up}}$ 
denotes the time for a single vocabulary distribution transmission.

\subsection{Wall-clock Time}
Inference latency comprises three parts: on-device SLM drafting time, edge-side LLM verification time, and device–edge communication time.

In a single run of the Draft-Verify process, the inference latency is:
\begin{equation}
\label{eq:T_inf}
\begin{aligned}
    T_{inf} &= \gamma \cdot T_{SLM} + T_{comm} + T_{LLM} \\
    &= \gamma \cdot (T_{SLM}+T_{\mathcal{V}}) + T_{LLM}
\end{aligned}
\end{equation}
where $T_{SLM}$ and $T_{LLM}$ denote the time for a single run of SLM and LLM, respectively. 

Let
\begin{align}
    b = T_{\mathcal{V}}/T_{LLM}, ~~~~~
    c = T_{SLM}/T_{LLM}.
\end{align}
Thus, we define
\begin{align}
    L = b +c,
\end{align}
representing the relative cost of vocabulary transmission and SLM inference with respect to one LLM run, which will be incorporated into the speculative decoding latency together with the LLM verification cost.
And \eqref{eq:T_inf} is rephrased as:
\begin{equation}
\label{eq:T_inf_bc}
    T_{inf} = [1+\gamma L] T_{LLM}.
\end{equation}

To realize acceleration, the per-token inference latency of the small model, combined with its transmission overhead, must remain below the per-token inference latency of the large model. Therefore, we have $L<1$.

\section{TK-SLT Scheme}
\label{sec:top-K}
As is shown in \eqref{eq:Vocabulary}, token generation involves uploading a distribution whose payload size is proportional to the dimensionality of the vocabulary space.  Such excessive uplink transmission leads to prohibitive communication latency, thereby reducing overall inference throughput. To reduce transmission cost, we propose our solution, the \textbf{Top-$\boldsymbol{K}$ Sparse Logits Transmission (TK-TRPT)} scheme. 

The underlying intuition can be articulated as follows: provided that a draft token undergoes the verification process outlined in Sec.~\ref{sec:DSD}, DSD guarantees consistency with a standalone LLM. In particular, the resulting token probability distribution produced by DSD is theoretically equivalent to that of the original LLM, irrespective of the specific form or variation of $Q(x)$, the output distribution from SLM. Moreover, if $Q(x)$ is sufficiently sparse, the associated communication cost can be reduced significantly, without affecting the probabilistic equivalence with the original LLM.

\begin{solution}[\textbf{Top-$\boldsymbol{K}$ Sparse Logits Transmission, TK-SLT}]
The SLM applies softmax only to the $K$ largest logits, and these $K$ logits along with their corresponding token IDs are transmitted to the LLM for verification.
\end{solution}

In this scheme, only the top-$K$ logits are retained, resulting in a sparse representation. By transmitting this sparse set rather than the full logits, the communication overhead is significantly reduced, while the most probable candidate tokens are preserved for verification by the LLM. Although it still outputs a token probability distribution identical to that of the standalone LLM, the acceptance rate is modified. 

According to~\cite{leviathan2023fast}, the acceptance rate is
\begin{equation}
\alpha = E \left[\beta_i\right],~~~ \beta_i = \sum_x \min \left\{ q_i(x), p_i(x) \right\},
\end{equation}
where $Q_i(x)$ and $q_i(x)$ denote the distribution obtained from standard sampling (temperature=1). 

When Top-$K$ sampling method is employed, the acceptance rate becomes
\begin{equation}
\alpha^{(K)}=E \left[\beta_i^{(K)}\right] , ~~~ \beta_i^{(K)} = \sum_x \min \left\{ y^{(K)}_i(x), p_i(x) \right\},
\end{equation}
where $Y^{(K)}_i(x)$ and $y^{(K)}_i(x)$ denote the distribution obtained from Top-$K$ sampling.
The change in $Y^{(K)}_i(x)$ relative to $Q_i(x)$ is not deterministic and can either increase or decrease, depending on the relationship between $P_i(x)$ and $Q_i(x)$. Specifically:
\begin{itemize}
    \item If the top-$K$ points of \(Q_i(x)\) (i.e., those with the highest \(q_i(x)\) values) correspond to high values of \(P_i(x)\), then \(\alpha^{(K)}\) may increase because \(Y_i(x)\) concentrates probability mass on points where \(\min\left\{y_i^{(K)}(x),p_i(x)\right\}\) is large.
    \item Conversely, if the top-$K$ points of \(Q_i(x)\) correspond to low values of \(P_i(x)\), then \(\alpha^{(K)}\) may decrease because \(\min\left\{y^{(K)}_i(x), p_i(x)\right\}\) is small despite the concentration of mass.
\end{itemize}

Thus, the transformation of the value $\alpha$ to $\alpha^{(K)}$ following a top-$K$ operation on the distribution $Q_i(x)$ is not monotonic. The direction and magnitude of change are contingent upon the intrinsic relationship between the two probability distributions $P(x)$ and $Q_i(x)$. The variation of $\alpha$ has a significant impact on the speedup ratio, which will be analyzed in more detail in Section~\ref{sec:speedup}.

\section{Theoretical Analysis}
\label{sec:theoretical_analysis}
In this section, we introduce the speedup ratio and derive a closed-form characterization of the optimal draft length, which involves the Lambert $W$ function.
\subsection{Speedup ratio}
\label{sec:speedup}
To better analyze the system performance, we define the speedup ratio--$S_{\text{inf}}$ 
as follows:
\begin{equation}
    S_{\text{inf}} = \frac{\Theta_{\text{DSD}}}{\Theta_{\text{LLM}}},
\end{equation}
where $\Theta_{\text{DSD}}$ and $\Theta_{\text{LLM}}$ denote the inference throughput of distributed speculative decoding and that of standalone LLM, respectively. 

\begin{theorem}
The speedup ratio of DSD is 
\begin{equation}
    S_{ \text{inf} }(\gamma) = \frac{1-\alpha^{\gamma+1}}{[1+\gamma L](1-\alpha)},
\end{equation}
where $L = b+c$.
\end{theorem}
\begin{proof}
According to \cite{leviathan2023fast}, the expected token generated by DSD is
\begin{equation}
E(\#\text{tokens}) = \frac{1-\alpha^{\gamma+1}}{1-\alpha},
\end{equation}
where $\alpha$ is the expected acceptance rate.

Combining with \eqref{eq:T_inf_bc}, we have
\begin{equation}
    \Theta_{\text{DSD}} = \frac{1-\alpha^{\gamma+1}}{1-\alpha}\cdot \frac{1}{[1+\gamma L] T_{LLM}}.
\end{equation}

Since the inference throughput of standalone LLM is the reciprocal of $T_{LLM}$:
\begin{equation}
    \Theta_{\text{LLM}} = \frac{1}{T_{\text{LLM}}},
\end{equation}
we have 
\begin{equation}
    S_{\text{inf}}(\gamma) = \frac{1-\alpha^{\gamma+1}}{[1+\gamma L](1-\alpha)}.
\end{equation}
\end{proof}

\subsection{Optimal draft token length}
The optimal draft token length problem is formulated as:
\begin{align}
    \label{eq:max_S}
    \max_{\substack{\gamma}}~~& S_{\text{inf}}\\
    s.t.~~& 0<\alpha<1\\
    & \label{eq:gamma_cond}\gamma \in \mathbb{Z}, \gamma > 1\\
    & \label{eq:L} b+c<1\\
    & b>0,~~ c>0
\end{align} 
\begin{theorem}
\label{theorem:gamma}
The optimal draft length---$\gamma^{*}$ in \eqref{eq:max_S} takes $\lceil \gamma_0 \rceil$ or $\lfloor \gamma_0 \rfloor$ if $\gamma_0>1$, and otherwise, $\gamma^{*}=1$, where $\lceil \cdot \rceil$ and $\lfloor \cdot \rfloor$ are the ceiling function and floor function, respectively. And
\begin{align}
\label{eq:gamma_0}
\gamma_0 =
\frac{1}{ \ln \alpha } \left[W_{-1}\left( -\frac{1}{e}\alpha^{\frac{1}{L}-1} \right) + 1 \right] - \frac{1}{L},
\end{align}
where $W_{-1}(\cdot)$ denotes the \textbf{Lambert~$\boldsymbol{W}$} function ($-1$ branch), and $L=b+c$.
\end{theorem}
\begin{proof}
\eqref{eq:max_S} is equivalent as
\begin{equation}    \max_{\substack{\gamma}}~~ T(\gamma) =\frac{1-\alpha^{\gamma+1}}{1+L\gamma}.
\end{equation}

\noindent~\textbf{Step 1: First Derivative and Critical Point}

The first derivative of $T(\gamma)$ is
\begin{equation}
    T'(\gamma) = \frac{ -\alpha^{\gamma+1} \ln \alpha \cdot (1 + L\gamma) - L\left(1 - \alpha^{\gamma+1}\right) }
{(1+L\gamma)^2}.
\end{equation}

Critical points occur when the numerator vanishes:
\begin{equation}
\label{eq:numerator_0}
    -\alpha^{\gamma+1} \ln \alpha \cdot (1 + L\gamma) - L\left(1 - \alpha^{\gamma+1}\right) = 0.
\end{equation}

\noindent~\textbf{Step 2: Variable Substitution}

Let 
\begin{equation}
\label{eq:k}
    k = \ln \alpha,
\end{equation}
where $k < 0$ since $\alpha < 1$.

And substituting \eqref{eq:k} into \eqref{eq:numerator_0}, we have
\begin{equation}
\label{eq:Numerator_0}
    \alpha^{\gamma} \left[ k(1 + L\gamma) - L \right] = - \frac{L}{\alpha}.
\end{equation}
Let
\begin{equation}
\eta = \gamma - \frac{1}{k} + \frac{1}{L}.    
\end{equation}
Then
\begin{equation}
\label{eq:gamma_sub}
\gamma = \eta + \frac{1}{k} - \frac{1}{L}.
\end{equation}
Substituting \eqref{eq:gamma_sub} into \eqref{eq:numerator_0}, we have
\begin{equation}
\label{eq:complex}
\alpha^{\eta + \frac{1}{k} - \frac{1}{L}} k \eta = -\frac{1}{\alpha}.
\end{equation}
Since $k=\ln \alpha$, we have $\alpha^{\frac{1}{k}}=e$. Substituting it into \eqref{eq:complex}, we have
\begin{equation}
\label{eq:lambert1}
    \eta e^{k\eta} = -\frac{1}{ek\alpha^{1-\frac{1}{L}}}.
\end{equation}

\noindent~\textbf{Step3: Labmert~$W$ Function Transformation}

Let 
\begin{equation}
\label{eq:w}
    w = k\eta.
\end{equation} 
We transform \eqref{eq:lambert1} into the equation:
\begin{equation}
\label{eq:labmert_W}
    w e^w = -\frac{1}{e} \alpha^{\frac{1}{L}-1}
\end{equation}

According to \eqref{eq:L}, $\frac{1}{L}-1>0$. Hence, $\alpha^{\frac{1}{L}-1}<1$, which further implies
$-\frac{1}{e} \alpha^{\frac{1}{L}-1} \in \left( -\frac{1}{e}, 0 \right)$. This places the argument in the domain where two real branches of the \textbf{Lambert~$\boldsymbol{W}$} function are well defined, i,e. $W_0(\cdot)$ and $W_{-1}(\cdot)$. Therefore, the solution to \eqref{eq:labmert_W} can be expressed as
\begin{equation}
\label{eq:Lambert_W_solution}
    w = W_{i} \left(-\tfrac{1}{e}\,\alpha^{\tfrac{1}{L}-1}\right), ~~i=0~~or~-1
\end{equation}

\noindent~\textbf{Step4: Solve for Optimal $\boldsymbol{\gamma^*}$}

Reverting to $\gamma$, we substitute \eqref{eq:k}, \eqref{eq:w} and \eqref{eq:Lambert_W_solution} into \eqref{eq:gamma_sub}:
\begin{align}
\gamma &= \frac{w}{k} + \frac{1}{k} - \frac{1}{L} \notag\\
       &\overset{(a)}{=} \frac{1}{\ln\alpha}(w+1)-\frac{1}{L} \notag\\
       &\overset{(b)}{=} \frac{1}{\ln\alpha}\left[W_{i} \left(-\tfrac{1}{e}\,\alpha^{\tfrac{1}{L}-1}\right)+1\right]-\frac{1}{L}
\end{align}
where $(a)$ is from \eqref{eq:k}, and $(b)$ is from \eqref{eq:w} and \eqref{eq:Lambert_W_solution}.

Since $W_{0}(\cdot) \in [-1, 0)$, and $\ln\alpha<0$, we have $\gamma \leq -\frac{1}{L}$. This contradicts \eqref{eq:gamma_cond}. Hence, we have
\begin{equation}
    \gamma_0=\frac{1}{\ln\alpha}\left[W_{-1} \left(-\tfrac{1}{e}\,\alpha^{\tfrac{1}{L}-1}\right)+1\right]-\frac{1}{L}
\end{equation}
Since $\gamma^*$ represents the optimal draft length in DSD, it must be an integer greater than or equal to 1, which leads to Theorem.~\ref{theorem:gamma}.
\end{proof}

\section{Practical Implementation}
\subsection{Optimal Draft Length Determination}
According to Theorem~\ref{theorem:gamma}, we provide in Algorithm~\ref{alg:gamma} the procedure for determining the optimal draft length: first compute $\gamma_0$ according to \eqref{eq:gamma_0}. Then, set $\gamma^*=1$ if $\gamma_0 < 1$. Otherwise, compute and compare $S_{\text{inf}}(\lceil \gamma_0 \rceil)$ and $S_{\text{inf}}(\lfloor \gamma_0 \rfloor)$. Finally, choose the value $\lceil \gamma_0 \rceil$ or $\lfloor \gamma_0 \rfloor$ that maximizes $S_{\text{inf}}$.

\begin{algorithm}[ht]
\caption{Optimal Draft Length Determination (ODLD)}
\label{alg:gamma}
\begin{algorithmic}
\STATE {\bfseries Input:} $\alpha$, $b$, $c$
\STATE Compute $\gamma_0 = \frac{1}{ \ln \alpha } \left[W_{-1}\left( -\frac{1}{e}\alpha^{\frac{1}{b+c}-1} \right) + 1 \right] - \frac{1}{b+c}$\\
\IF{$\gamma_0 < 1$}
\STATE $\gamma^*=1$
\ELSE 
\STATE  Compute $S_{\text{inf}}(\lceil \gamma_0 \rceil)$ and $S_{\text{inf}}(\lfloor \gamma_0 \rfloor)$, respectively.
    \IF{$S_{\text{inf}}(\lfloor \gamma_0 \rfloor) \leq S_{\text{inf}}(\lceil \gamma_0 \rceil)$}
    \STATE $\gamma^* = \lceil \gamma_0 \rceil$;\\ $S_{\text{inf}}^* = S_{\text{inf}}(\lceil \gamma_0 \rceil)$
    \ELSE 
    \STATE $\gamma^* = \lfloor \gamma_0 \rfloor$;\\
    $S_{\text{inf}}^* = S_{\text{inf}}(\lfloor \gamma_0 \rfloor)$
    \ENDIF
\ENDIF
\STATE {\bfseries Output:} $\gamma^*$, $S_{\text{inf}}^*$
\end{algorithmic}
\end{algorithm}

\subsection{Adaptive Speculative Selection}
Although the optimal draft length $\gamma^*$ can be computed according to Theorem~\ref{theorem:gamma}, it must still be ensured that DSD yields performance gains over standalone LLM inference. That is $S_{\text{inf}}(\gamma^*) >1$. Based on this consideration, we design an adaptive speculative selection mechanism, which is presented in Algorithm~\ref{alg:spec_selection}.
\begin{algorithm}[ht]
\caption{Adaptive Speculative Selection (AS²)}
\label{alg:spec_selection}
\begin{algorithmic}[1]
\STATE \textbf{Input:} $\alpha$, $b$, $c$
\STATE $\gamma^*, S_{\text{inf}}^* \gets$ Call \textbf{Algorithm 1: ODLD}($\alpha$, $b$, $c$)
\IF{$S_{\text{inf}}^*>1$}
    \STATE Employ DSD
\ELSE
    \STATE Employ standalone LLM
\ENDIF
\end{algorithmic}
\end{algorithm}


\section{Numerical Results and Simulations}
\label{sec:experiment}
\subsection{Numerical Analysis on $S_{\text{inf}}$ and $\gamma^*$}

\begin{figure*}[ht]
    \centering
    \subfloat[$\alpha = 0.4$]{%
        \includegraphics[width=0.3\textwidth]{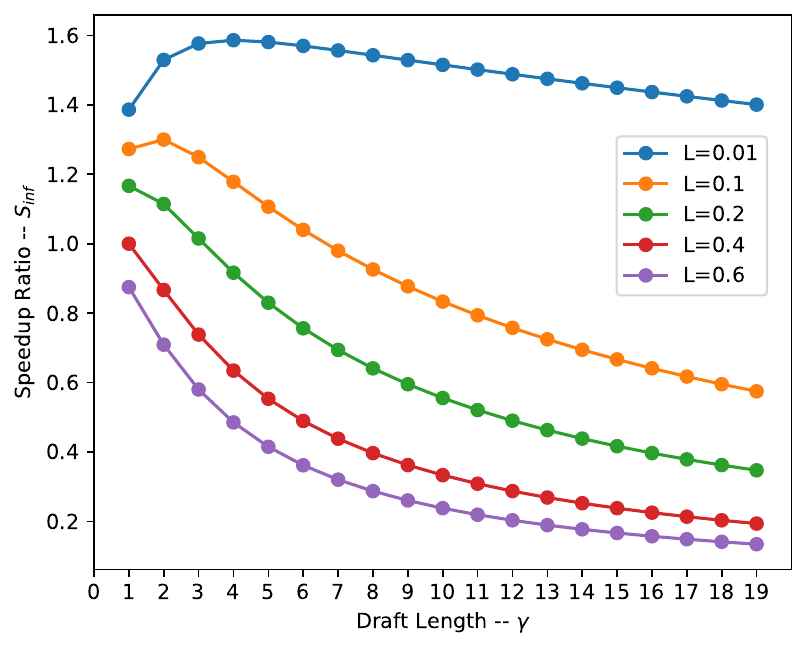}
        \label{fig:sub_a}
    }
    \hfill
    \subfloat[$\alpha = 0.6$]{%
        \includegraphics[width=0.3\textwidth]{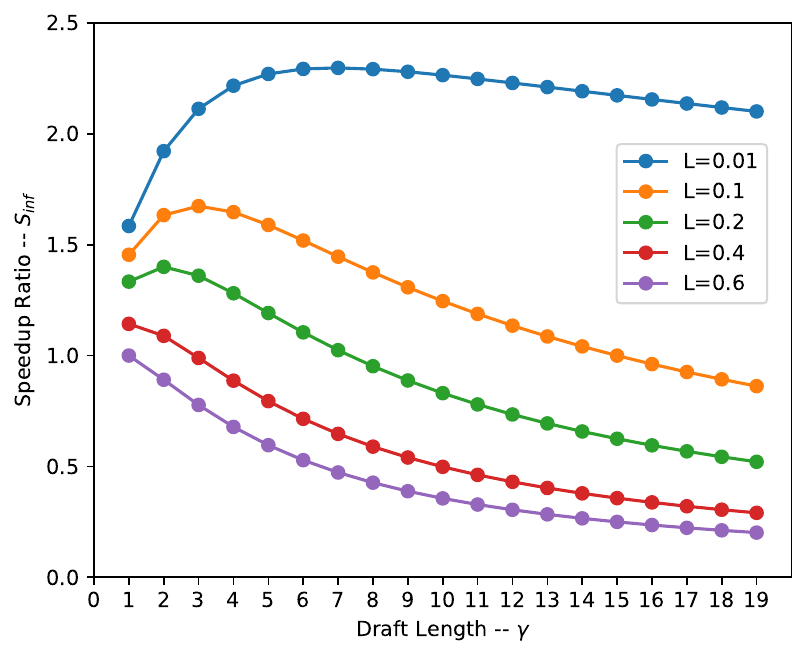}
        \label{fig:sub_b}
    }
    \hfill
    \subfloat[$\alpha = 0.8$]{%
        \includegraphics[width=0.3\textwidth]{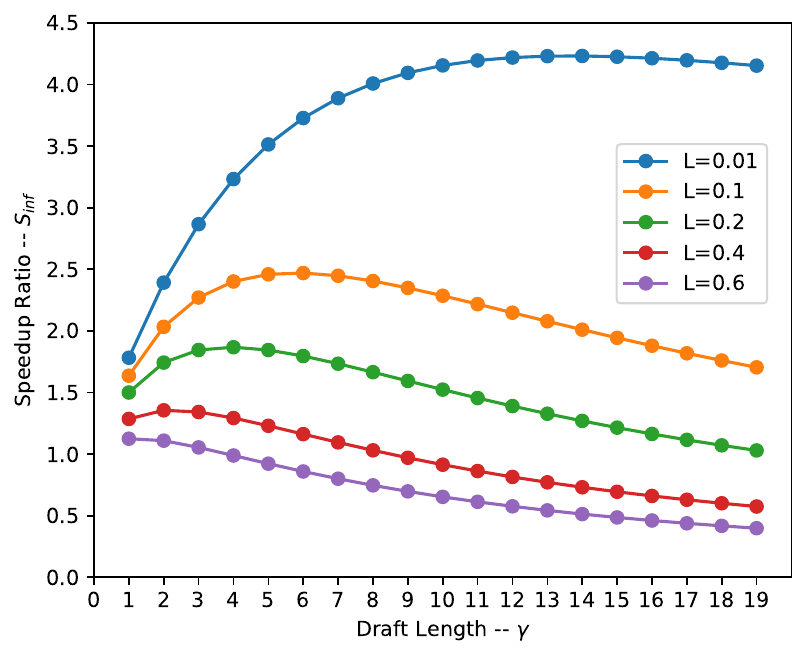}
        \label{fig:sub_c}
    }
    \caption{Effect of draft length on speedup ratio at different $L$}
    \label{fig:speedup}
\end{figure*}

The effect of draft length on speedup ratio at different $L$ is shown in Fig~\ref{fig:speedup}. It can be observed that when $\alpha$ is large and $L$ is small, there exists an optimal draft length that maximizes $S_{\text{inf}}$. This is because, at the beginning, $\alpha$ is relatively small, and thus the generated tokens are very likely to be fully accepted. However, as the draft length increases, more tokens are rejected during the verification process. Meanwhile, a longer draft length also introduces higher communication overhead and additional inference latency at the SLM. On the other hand, when $\alpha$ is small and $L$ is large, $S_{\text{inf}}$ decreases as the draft token length $\gamma$ increases. This is because, in this case, the probability of accepting draft tokens is already low, and extending the draft length only amplifies the rejection ratio. In addition, a larger $\gamma$ further aggravates the communication overhead and the SLM inference delay, leading to a monotonic degradation of the speedup ratio. It is also worth noting that in some cases $S_{\text{inf}} < 1$ when $\gamma=1$, which implies that speculative decoding is even less efficient than standalone LLM inference. In such scenarios, it is preferable to directly adopt the standalone LLM rather than involving the SLM.

Meanwhile, Table~\ref{tab:gamma_values} presents the optimal $\gamma^*$ for different values of $\alpha$ and $L$, obtained using Algorithm~\ref{alg:gamma}. The entries marked with $1^*$ indicate the cases where $S_{\text{inf}}<1$, in which standalone LLM inference should be employed. These observations are consistent with the trends illustrated in Fig.~\ref{fig:speedup}.

\begin{table}[h]
\centering
\caption{Optimal $\gamma^*$ for different values of $\alpha$ and $L$}
\label{tab:gamma_values}
\begin{tabular}{c|ccccc}
\hline
$\alpha \backslash L$ & 0.01 & 0.1 & 0.2 & 0.4 & 0.6 \\ \hline
0.4 & $4$ & $2$ & $1$ & $1$ & $1^*$\\
0.6 & $7$ & $3$ & $2$ & $1$ & $1^*$\\
0.8 & $14$ & $6$ & $4$ & $2$ & $1$\\
\hline
\end{tabular}
\end{table}



\subsection{Hardware-in-the-Loop Simulation for TK-SLT}

Two models are deployed on separate NVIDIA A800 80GB GPUs: a \textbf{68M-Llama} model for draft generation and a \textbf{7B-Llama} model for verification. The vocabulary size is configured as $32$K~\cite{touvron2023llama}. 

We set the temperature to $T =0,1$ for LLM inference, corresponding to deterministic greedy decoding, and stochastic sampling, respectively. For SLM inference, we employ top-$K$ sampling with a fixed temperature of $1$.

\begin{figure}[htbp]
\begin{center}
\centerline{\includegraphics[width=0.9\columnwidth]{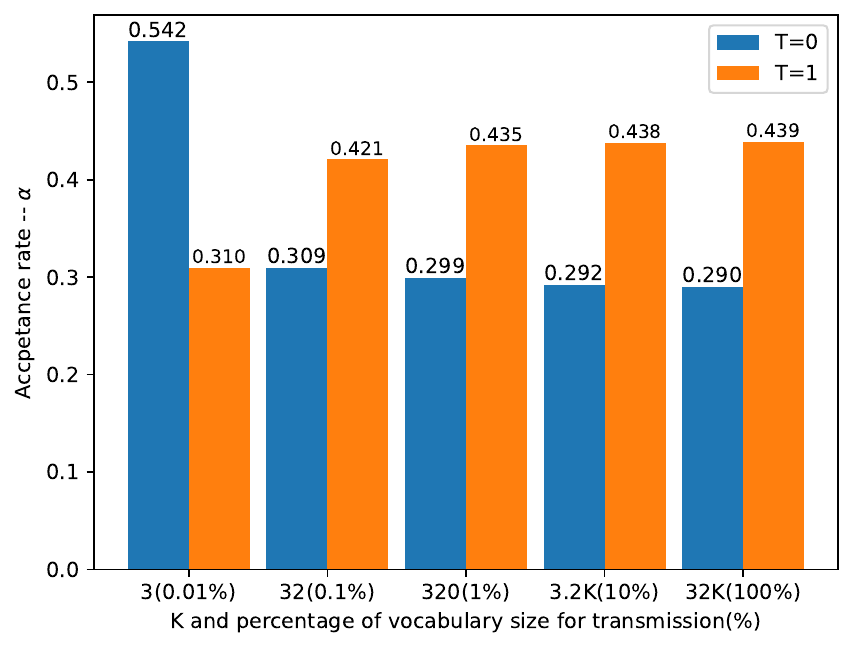}}
\caption{Acceptance rate under different $K$}
\vspace{-8mm}
\label{fig:acceptance_rate}
\end{center}
\end{figure}

The estimated $\alpha$ across various values of $K$ for top-$K$ sampling and transmission is presented in Fig.~\ref{fig:acceptance_rate}.
When the target LLM uses $T=0$, LLM inference reduces to greedy decoding, producing an extremely peaked token distribution that concentrates probability mass on the top-ranked token. Consequently, restricting SLM sampling to a small top-$K$ similarly limits the output support to high-probability tokens, yielding comparable inference behavior. In this regime, smaller $K$ values consistently lead to larger $\alpha$.
By contrast, increasing $T$ flattens the target distribution, which generally reduces the overlap between draft proposals and target-preferred tokens and therefore $\alpha$ tends to be smaller when $K$ is small. 
Empirically, we observe that the negative dependence of $\alpha$ on $K$ is strongest at $T=0$, whereas at $T=1$ the trend reverses, with $\alpha$ increasing as $K$ grows.

In addition to the empirical evaluation, we simulate a communication mode where logits are quantized to half-precision (FP16) and transmitted at $50$Mbps\cite{gsm2023estimating}. 
Through our simulation, we obtained the following estimation: $b \approx 0.07$ and $c \approx 0.23$ for full logits transmission.

Table~\ref{tab:top_K} shows that $L$ decreases markedly as the transmission cost drops from $K=32{,}000$ to $K=320$. In contrast, when $K$ is further reduced from $320$ to $3$, $L$ remains essentially constant, since the transmission cost is already negligible. The change in $K$ simultaneously induces adjustments in the optimal draft length~$\gamma^*$.

\begin{table}[ht]
\centering
\caption{$L$ and Optimal $\gamma^*$ under different $K$}
\label{tab:top_K}
\begin{tabular}{c|ccccc}
\hline
$K$ & 3 & 32 & 320 & 3200 & 32000 \\
\hline
$L$ & 0.0700 & 0.0702 & 0.0723 & 0.093 & 0.300 \\
\hline
$\gamma^* (T=0)$ & 3  & 2 & 2 & 1 & 1 \\
\hline
$\gamma^* (T=1)$ & 2  & 2 & 2 & 2 & 1  \\
\hline
\end{tabular}
\end{table}

Fig.~\ref{fig:speedup_ratio} illustrates the speedup ratio $S_{\text{inf}}$ under different values of $K$. When $T=0$, $S_{\text{inf}}$ attains its maximum since the acceptance rate $\alpha$ is highest and the transmission cost is minimal. Nevertheless, when jointly considering both $\alpha$ and the transmission overhead, the optimal speedup is achieved at $K=320$.

\begin{figure}[htbp]
\begin{center}
\centerline{\includegraphics[width=0.9\columnwidth]{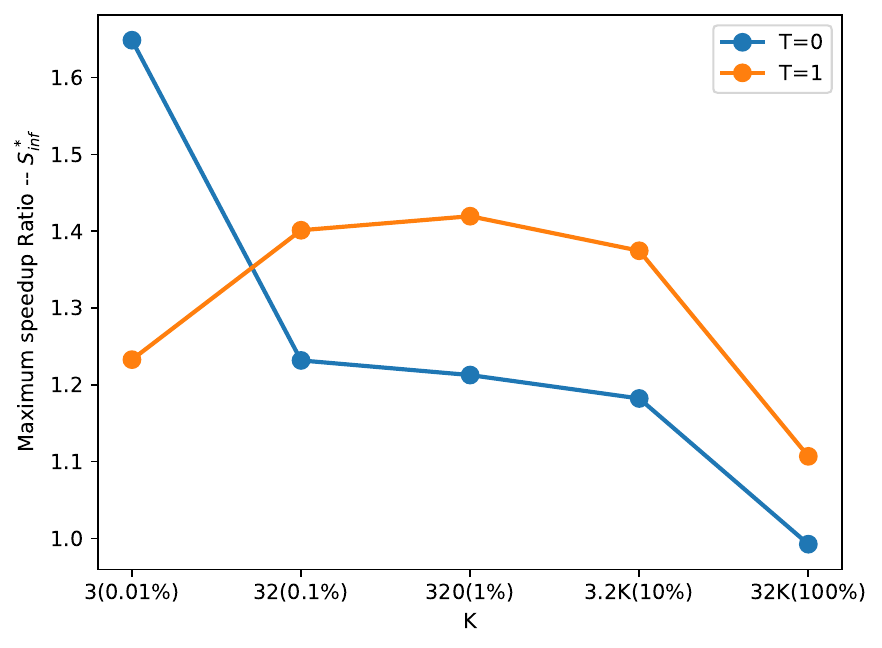}}
\caption{Speedup ratio $S_{\text{inf}}$ under different $K$}
\vspace{-6mm}
\label{fig:speedup_ratio}
\end{center}
\end{figure}

\section{Conclusion}
\label{sec:conclusion}
In this paper, we have proposed the \textbf{Top-$K$ Sparse Logits Transmission (TK-SLT)} scheme to alleviate the prohibitive communication overhead in distributed speculative decoding. By transmitting only the most probable logits instead of the entire vocabulary distribution, TK-SLT achieves substantial reduction in uplink payload size while preserving probabilistic equivalence with the standalone LLM. We have further established a rigorous theoretical framework to characterize the inference speedup, deriving the optimal draft length in closed form using the Lambert-$W$ function. The resulting algorithms---Optimal Draft Length Determination (ODLD) and Adaptive Speculative Selection (AS\textsuperscript{2})---enable practical implementation under diverse system parameters. Numerical analysis and hardware-in-the-loop experiments have confirmed the effectiveness of TK-SLT. 



\section{Acknowledge}
This work was supported by the Guangdong S\&T Programme under Grant No. 2024B0101010003, in part by Major Key Project of PCL under Grant No.PCL2025AS209, and in part by National Key Research and Development Program of China under Grant No.2024YFE0200800.

\bibliographystyle{IEEEtran}
\bibliography{main}

\end{document}